\documentclass[sigconf,screen]{acmart}
\usepackage[normalem]{ulem}

\usepackage{amsmath,geometry,arydshln,subcaption,pifont,algpseudocode,algorithm,xcolor,blkarray,enumitem,cleveref,bm}
\newcommand{\cmark}{\ding{51}}%
\newcommand{\xmark}{\ding{55}}%
\algnewcommand\algorithmicforeach{\textbf{for each}}
\algdef{S}[FOR]{ForEach}[1]{\algorithmicforeach\ #1\ \algorithmicdo}

\crefname{section}{§}{§§}
\Crefname{section}{§}{§§}

\acmYear{2024}\copyrightyear{2024}
\setcopyright{acmlicensed}
\acmConference[Middleware '24]{25th International Middleware Conference}{December 2--6, 2024}{Hong Kong, Hong Kong}
\acmBooktitle{25th International Middleware Conference (Middleware '24), December 2--6, 2024, Hong Kong, Hong Kong}
\acmDOI{10.1145/3652892.3654792}
\acmISBN{979-8-4007-0623-3/24/12}

\graphicspath{{images/}{../images/}}

\begin{document}

\title{Optimal Resource Efficiency with Fairness in Heterogeneous GPU Clusters}

\author{Zizhao Mo}
\affiliation{
	\institution{University of Macau}
	\country{Macau SAR, China}
}
\email{yc17461@um.edu.mo}

\author{Huanle Xu}
\affiliation{
	\institution{University of Macau}
	\country{Macau SAR, China}
}
\email{huanlexu@um.edu.mo}
\authornote{Corresponding author}

\author{Wing Cheong Lau}
\affiliation{
	\institution{The Chinese University of Hong Kong}
	\country{Hong Kong}
}
\email{wclau@ie.cuhk.edu.hk}



\begin{abstract}

Ensuring the highest training throughput to maximize resource efficiency, while maintaining fairness among users, is critical for deep learning (DL) training in heterogeneous GPU clusters. However, current DL schedulers provide only limited fairness properties and suboptimal training throughput, impeding tenants from effectively leveraging heterogeneous resources. The underlying design challenge stems from inherent conflicts between efficiency and fairness properties.

In this paper, we introduce OEF, a new resource allocation framework specifically developed for achieving optimal resource efficiency and ensuring diverse fairness properties in heterogeneous GPU clusters. By integrating resource efficiency and fairness within a global optimization framework, OEF is capable of providing users with maximized overall efficiency, as well as various guarantees of fairness, in both cooperative and non-cooperative environments. We have implemented OEF in a cluster resource manager and conducted large-scale experiments, showing that OEF can improve the overall training throughput by up to 32\% while improving fairness compared to state-of-the-art heterogeneity-aware schedulers.
\end{abstract}

%\makeatletter
%\let\@authorsaddresses\@empty
%\makeatother
%\settopmatter{printacmref=false}
%\renewcommand\footnotetextcopyrightpermission[1]{}
% \pagestyle{plain}

\maketitle

\section{Introduction}
\label{introduction}

Modern deep-learning (DL) tasks are rapidly being applied across various domains, including object detection~\cite{he2017mask}, natural language processing~\cite{vaswani2017attention}, and video processing~\cite{shafiee2017fast}. To meet the increasing demand for training DL tasks, major companies have established multi-tenant GPU clusters where multiple users need to share GPU resources. Additionally, production clusters today feature a variety of accelerator types due to the rapid pace of accelerator updates. For example, in Google Cloud Platform~\cite{googlegpu}, seven types of GPUs coexist: A100, V100, P100, K80, T4, P4, and L4. While heterogeneous clusters provide users with more choices in terms of pricing and efficiency trade-offs for DL training, the complex runtime profiles of users—specifically, diverse training throughput across GPU devices—pose significant challenges in achieving high resource efficiency~\cite{tiresias, OptimusDL, le2020allox, gavel} and ensuring fairness~\cite{Gandivafair, gavel, themis} in resource allocation among users.

From an operator's perspective, it is always imperative to provide users with the highest overall training throughput for the best resource efficiency. DL jobs typically have significantly different speedups of training throughput between GPU devices~\cite{gavel, le2020allox, Gandivafair}. Specifically, the speedup ranges from almost no speedup to several times due to architectural differences between DL models. For high resource efficiency, it is desirable to allocate jobs with the most speedup on high-end GPUs to produce the highest training throughput. However, this can easily lead to extremely unbalanced efficiency, i.e., jobs with great speedup monopolize the fast GPU whereas other jobs have to be starved or assigned to low-end devices. As a consequence, simply optimizing resource efficiency grants negative \textit{sharing incentive} to jobs with low speedup, and encourages users to cheat on their job speedup rate to win allocations on fast GPUs. In this sense, pure efficiency maximization shall violate \textit{strategy-proofness}, which is an important fairness requirement as it can prevent users from lying about their requirements to improve their training performance.

%From a user's perspective, fairness has always been an important issue. A big challenge herein is that existing fairness schedulers are no-longer suitable for heterogeneous GPU clusters. First, the widely adopted single-resource fairness principle, known as \textit{max-min fairness}~\cite{max-min-fairness, liu1973scheduling, feitelson2015workload}, which evenly distributes GPU shares, can negatively impact overall resource efficiency due to the diverse speedup among different jobs. Second, GPUs are considered \emph{interchangeable} resources for DL jobs, as explained in \cref{heterogeneity}. Therefore, we cannot directly apply existing multi-resource fairness schedulers~\cite{DRF, quincy, carbyne, hug, khamse2017per, wang2014multi, wang2016multi} that allocate non-interchangeable resources (such as CPU, memory, or network) to heterogeneous GPU scheduling. In GPU clusters, it remains an open problem to design an efficient and fair policy that can offer users multiple fairness properties, including \textit{envy-freeness} (where users do not prefer others' allocations over their own) and \textit{sharing-incentive} (which incentivizes users to share resources), in addition to \textit{strategy-proofness}~\cite{gavel}.
From a user's perspective, fairness has always been an important issue. A big challenge herein is that existing fairness schedulers are no-longer suitable for heterogeneous GPU clusters. First, the widely adopted fairness principle \textit{max-min fairness}~\cite{max-min-fairness, liu1973scheduling, feitelson2015workload} and its variants in operating system~\cite{baruah1993proportionate, baruah1995fast, zhu2003multiple} and network domains~\cite{blanquer2001fair, kleinberg1999fairness} are tailored for resources of the same type. While this principle aims to equitably distribute resources, it can have an adverse impact on overall resource efficiency due to variations in speedup levels among different jobs. Second, GPUs are considered \emph{interchangeable} resources for DL jobs, as explained in \cref{heterogeneity}. Therefore, we cannot directly apply existing multi-resource fairness schedulers~\cite{DRF, quincy, carbyne, hug, khamse2017per, wang2014multi, wang2016multi} that allocate non-interchangeable resources (such as CPU, memory, or network) to heterogeneous GPU scheduling. In GPU clusters, it remains an open problem to design an efficient and fair policy that can offer users multiple fairness properties, including \textit{envy-freeness} (where users do not prefer others' allocations over their own) and \textit{sharing-incentive} (which incentivizes users to share resources), in addition to \textit{strategy-proofness}~\cite{gavel}.

In this paper, our focus is to address a fundamental problem: maximizing overall resource efficiency while meeting various fairness requirements for scheduling with heterogeneous GPUs. Existing researches have proposed efficient scheduling methods in this area~\cite{gavel, le2020allox, Gandivafair}, but these approaches primarily focus on training efficiency, leading to fairness violations. Furthermore, even schedulers that prioritize fairness can only guarantee \textit{sharing-incentive} and \textit{pareto-efficiency}~\cite{gavel,Gandivafair}. Among them, Gandiva$_{\mbox{fair}}$ shares a similar objective to ours~\cite{Gandivafair}. Specifically, it aims to strike a balance between efficiency and fairness by implementing a trading mechanism based on \textit{max-min} fairness~\cite{jaffe1981bottleneck}. However, Gandiva$_{\mbox{fair}}$ is far from optimal as it fails to address cheating or ensure \textit{envy-freeness}.

We have identified the underlying reasons why current schedulers fall short in delivering adequate fairness properties and optimal efficiency. First, from a methodological perspective, these schedulers tightly couple their allocation schemes to specific fairness properties, attempting to integrate other properties in a best-effort manner without global coordination. Second, in principle, our analysis reveals inherent conflicts among different fairness properties in heterogeneous GPU scheduling. Pursuing high efficiency, for instance, proves incompatible with achieving \textit{strategy-proofness}, \textit{envy-freeness}, and \textit{sharing-incentive}. Simply prioritizing the latter two fairness properties can, in fact, result in more harm than good.

Based on the identification, we present OEF, a resource allocation framework capable of globally coordinating multiple fairness properties and maximum resource efficiency in the context of heterogeneity. At its core, OEF explicitly quantifies efficiency and fairness, and formulates optimization problems to explore the best coordination between them in both cooperative and non-cooperative environments. These environments have been extensively studied in the scheduling literature to address fairness concerns~\cite{hug}. In non-cooperative environments, such as multi-tenant platforms where individuals or entities need to compete for limited resources without a pricing mechanism, ensuring \textit{strategy proofness} is imperative. Non-cooperative OEF incentivizes users to reveal the true training throughput in their DL jobs across various devices, facilitating optimized resource efficiency concurrently. In cooperative environments where \textit{strategy-proofness} might not be a rigid requirement, OEF simultaneously achieves \textit{envy-freeness} and \textit{sharing-incentive} while guaranteeing optimal efficiency.

Moreover, OEF demonstrates its extensibility by supporting complex scheduling scenarios, such as accommodating users with varying priorities and enabling a user to run multiple types of DL tasks simultaneously. Furthermore, OEF only assigns GPUs with adjacent speedups to each user, thereby mitigating the adverse effect of skewed training speed across GPU types. We implemented a prototype system of OEF in a local cluster comprising 24 heterogeneous GPU devices to assess its performance. The results demonstrate that OEF successfully achieves desired fairness properties while outperforming state-of-the-art baselines. It provides users with a higher overall training throughput in both non-cooperative and cooperative settings, with improvements of up to 32\%. Additionally, OEF significantly reduces the overall job completion time in the long term, achieving reductions of up to 19\%. In summary, we have made the following contributions in this paper:

\begin{itemize}
    \item We identify the  limitations of existing fairness frameworks tailored specifically for heterogeneous GPU clusters and highlight the inherent conflicts between resource efficiency and diverse fairness properties in the context of heterogeneous GPU scheduling.
    \item We design OEF, the first heterogeneity-aware resource allocation framework that can achieve \textit{strategy-proofness} and \textit{envy-freeness} while pursuing high resource efficiency. 
    % \textcolor{red}{To our knowledge, OEF is the first schedulers that can offer these two properties.} 
    Additionally, OEF exhibits flexibility by accommodating users with varying priorities and  multiple types of training workloads.
    \item We demonstrate that OEF exhibits high resource efficiency, excellent scalability and introduces minimal computational overhead when identifying the optimal allocation in large-scale clusters.
\end{itemize}

\section{Background and Motivation}
\label{sec:motivation}
\subsection{Multi-tenant GPU Clusters}
\label{multitenantscheduling}

Today, GPU clusters are often shared among multiple tenants to reduce the operating costs and in the meanwhile, improve resource utilization~\cite{Gandivafair}. To this end, a large number of DL jobs from tenants need to be served in production GPU clusters~\cite{mlaasworkload,hu2021characterization}.  As a representative use case, to train an accurate model for a given task, one may first run a large number of related training jobs with different combinations of hyper-parameters, e.g., learning rate, batch size, or dropout rate, to explore the best combination that can yield the highest accuracy. Since these jobs are of similar model structure, they are (almost) equally accelerated given the same GPU. 
% {\color{red} Reported by~\cite{mlaasworkload}, about 90\% jobs are such recurring hyperparameter searching jobs in Alibaba cluster. Additionally, tenants may explore on different DL models simultaneously for the same tasks, where acceleration effect of jobs may differ diversely due to the achitecture difference. For this reason, schedulers should be compatible with such situation and function correctly to maximize the efficiency}. 
As reported by~\cite{mlaasworkload}, approximately 90\% of jobs in Alibaba clusters are recurring hyperparameter search jobs. Furthermore, tenants may concurrently experiment with different DL models for the same tasks, where the acceleration effect of jobs may vary significantly due to architectural differences. Therefore, schedulers should be adaptable to such situations and function effectively to maximize efficiency.

\begin{figure}
\begin{minipage}{0.49\linewidth}
\centering
\includegraphics[width=0.99\linewidth]{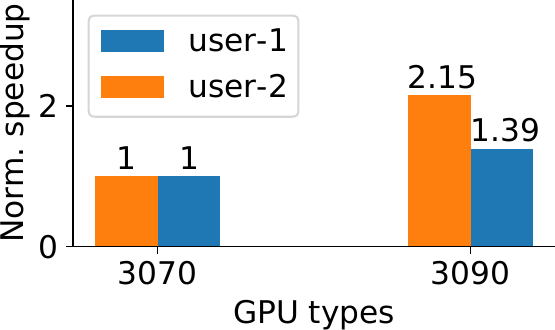}
%\vspace{-1.5em}
\subcaption{{Diverse speedup for users with different DL training jobs}}
% \vspace{-1.8em}
\end{minipage}
\begin{minipage}{0.49\linewidth}
\centering
\includegraphics[width=0.99\linewidth]{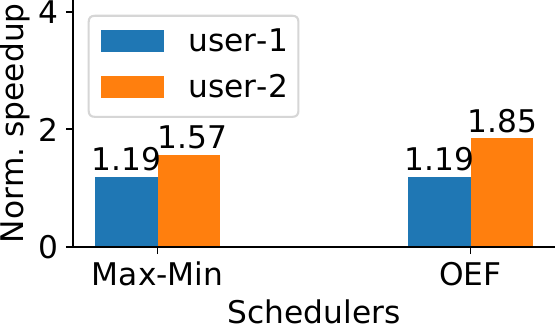}
%\vspace{-1.5em}
\subcaption{{Diverse speedup for users under different fair schedulers}}
% \vspace{-1.8em}
\end{minipage}
%\vspace{-.5em}
\caption{The effect of GPU heterogeneity in DL training clusters, speedup denotes the training throughput normalized by that on the slowest GPU type. In this illustration, user-1 executes a VGG model, while user-2 engages in running an LSTM model.}
\label{fig:motivation_hetero}
%\vspace{-1.5em}
\end{figure}

\subsection{Diverse Speedup under GPU Heterogeneity}
\label{heterogeneity}
The GPU heterogeneity in production clusters complicates scheduling among tenants. Specifically, DL training models typically show heterogeneous acceleration performance across GPU types due to architectural differences between models. Fig.~\ref{fig:motivation_hetero}(a) exhibits the skewed acceleration effect across jobs normalized by the training speed on the slowest GPU type. We can observe that VGG achieves $1.39\times$ speedup on RTX 3090. By contrast, LSTM achieves much higher speedups on more advanced GPUs, i.e., $2.15\times$ speedup on 3090.

\subsection{Fair Allocation in Heterogeneous Environment}
\label{fairnees-definition}
Given $k$ types of GPUs with $m_k$ devices each, and the training throughput of DL jobs on these $k$ types from $n$ different tenants (users), the cluster manager must use a scheme $\mathcal{S}$ to determine the allocation matrix $X$, which specifies the number of each type of GPU devices assigned to each tenant:
\begin{equation*}
\mathcal{S}\Big(\big\{\vec{W_1};\vec{W_2};\cdots;\vec{W_n}\big\}\Big) = X = \Big\{\vec{x_1};\vec{x_2};\cdots;\vec{x_n}\Big\}.
\end{equation*}
$\vec{W_l} = \left \langle w_l^1, w_l^2,\cdots, w_l^k\right\rangle$ represents the vector of training throughput for the $l$-th tenant. $\vec{x_l} = \left \langle x_l^1,x_l^2,\cdots,x_l^k\right\rangle$ is the allocation vector, here, $x_l^j$ denotes the number of devices of type $j$. To better formulate fairness, we normalize the throughput vector to the training throughput on the slowest GPUs, i.e., $w_l^1 = 1$ for all tenants \footnote{Due to hardware evolution, the slowest GPU type for different DL jobs is consistent, i.e. RTX 3070 is always slower than RTX 3090 when training various DL jobs.}. In this sense, $\vec{W_l}$ captures the speedup of different GPU types relative to the slowest type, also known as the speedup vector, and $W = \big\{\vec{W_1};\vec{W_2};\cdots;\vec{W_n}\big\}$ is the speedup matrix capturing all speedup information for all users.

\subsubsection{Fairness requirements}
As defined in previous works \cite{DRF,hug,themis}, a fair allocation scheme needs to meet four requirements:
\begin{enumerate}
\item
\noindent \textbf{\textit{Strategy-proofness}} ($\mathsf{SP}$): Tenants should not be able to benefit by lying about their speedups. In this sense, when the $l$-th tenant reports a fake vector $\vec{W_l}$ where each element is no less than that in the groundtruth, his attained normalized training throughput will decrease. This property serves as a deterrent, effectively preventing tenants from engaging in cutthroat competition for precious high-end GPU resources.

\item \textbf{\textit{Envy-freeness}} ($\mathsf{EF}$): A tenant should not prefer the allocation of another tenant, thereby ensuring equitable treatment for all tenants.

\item \textbf{\textit{Sharing-incentive}} ($\mathsf{SI}$): A tenant should obtain higher training throughput in a cluster partition consisting of $\frac{1}{n}$ of all types of GPU devices. Without this attribute, tenants would rather prefer an exclusive GPU cluster.
\item \textbf{\textit{Pareto-efficiency}} ($\mathsf{PE}$): It should not be possible to increase the training throughput of a user without decreasing the throughput of at least another user\footnote{We follow the original definition of PE given by DRF~\cite{DRF}. This definition is stronger than the one introduced in Gavel~\cite{gavel}, where the allocation is work conserving but throughput improvement can still be achieved via trading GPU shares among users.}. 
\end{enumerate}

\subsubsection{Optimal resource efficiency}
As demonstrated in extensive real-world studies~\cite{jeon2019analysis, mlaasworkload}, it is well-established that DL training jobs often span hours or even days, presenting significant opportunities for acceleration. This emphasizes the critical need to integrate resource efficiency in order to achieve optimized training throughput. In light of this requirement, the co-optimization of efficiency and fairness has the potential to further bolster the effectiveness of a scheduling scheme solely focused on fairness.

In particular, it is desirable to strive for the highest overall training throughput, referred to as \textit{optimal efficiency} while also taking various fairness considerations into account. It is important to clarify that this concept differs from \textit{pareto-efficiency}, as \textit{pareto-efficiency} allocations may not always lead to \textit{optimal efficiency}.

\subsubsection{Why need new fair schedulers?}
 \label{low-efficiency}
While Max-Min based schedulers excel in offering users various fairness attributes, they fall short in delivering satisfactory throughput performance to tenants in heterogeneous clusters, resulting in diminished overall resource efficiency. Illustrated in Fig.~\ref{fig:motivation_hetero}(b), User-2 experiences a training throughput decline of $1.18\times$ compared to that achieved under our proposed scheduler OEF, which concurrently maintains various fairness properties in a heterogeneity environment (elaborated in \cref{property-oef}). In aggregate, the cluster's overall efficiency contracts by approximately 10\%. This inefficiency arises because Max-Min fairness treats DL training jobs from tenants equally, inadequately leveraging high-end resources.

%The existence of different GPU types within a cluster raises the question of whether DRF-based schemes, which have been widely adopted for multi-resource allocation, can be applied to achieve fairness. DRF and its variants allocate resources proportionally to a user's requested resources, such as computation, storage, and networking. However, there is a significant distinction between our scenario and the multi-resource setting: GPUs of various types actually belong to the same category, namely computation resources. In other words, DL jobs can be executed on any type of GPU device, regardless of differences in training performance—an attribute known as \textit{interchangeability}~\cite{le2020allox}. On the contrary, a job that requires network resources cannot function if no network bandwidth is granted. Consequently, DRF-based schemes are not suitable for this new scenario, necessitating the development of innovative fairness solutions.
Another widely-adopted fairness schedulers, i.e., Dominant Resource Fairness (DRF)~\cite{DRF} and its variants~\cite{quincy, carbyne, hug, khamse2017per, wang2014multi, wang2016multi}, are also unfit for heterogeneous clusters. They primarily allocate resources proportionally to a user's requested resources, such as computation, storage, and networking. However, there is a significant distinction between our scenario and the multi-resource setting: GPUs of various types actually belong to the same category, namely computation resources. In other words, DL jobs can be executed on any type of GPU device, regardless of differences in training performance—an attribute known as \textit{interchangeability}~\cite{le2020allox}. On the contrary, a job that requires network resources cannot function if no network bandwidth is granted. Consequently, DRF-based schemes are not suitable for this new scenario, necessitating the development of innovative fairness solutions.

% subject to fairness constraints under scheme $\mathcal{S}$, 

\subsection{Limitation of Existing Heterogeneous Schedulers}
\label{inefficiencies}

\noindent\textbf{Gandiva$_{\mbox{fair}}$} designs a trading mechanism on top of max-min fairness. Specifically, Gandiva$_{\mbox{fair}}$ first equally divides GPU resources to users according to max-min fairness. Then, it performs greedy trading between users, that is, during the trading process, the fastest-accelerating users on high-end GPUs use a second-price auction mechanism to trade their shares of slow GPUs for shares of fast GPUs using the second price auction mechanism~\cite{vickrey}, where the second price denotes the speedup rate on fast GPUs of the second most accelerated user. Since this trading mechanism can only increases each user's training throughput, Gandiva$_{\mbox{fair}}$ satisfies \textit{sharing-incentive}. In addition, the second-price auction mechanism also prevents the fastest-accelerating user from lying about his speedup on fast GPUs. Consider an example where we have three users $u_1$, $u_2$, and $u_3$. The speedup matrix indicating the acceleration effect on diverse GPU types among jobs and the allocation matrix are presented as $W$ and $X$ respectively:
% \footnote{\textcolor{red}{Without special mention, users compete for two GPUs in our examples, one of type $\mbox{GPU}_1$ and one of type $\mbox{GPU}_2$.}}:
\begin{equation}
\label{efficieny_allo_example}
W =
\begin{blockarray}{ccc}
\mbox{GPU}_1 & \mbox{GPU}_2 & \\
\begin{block}{(cc)c}
1 & 2 & u_1 \\
1 & 3 & u_2 \\
1 & 4 & u_3 \\
\end{block}
&&
\end{blockarray},
X =
\begin{blockarray}{ccc}
\mbox{GPU}_1 & \mbox{GPU}_2 & \\
\begin{block}{(cc)c}
1 & 0.09 & u_1 \\
0 & 0.47 & u_2 \\
0 & 0.44 & u_3 \\
\end{block}
&&
\end{blockarray}.
%E = \begin{blockarray}{c}
%\begin{block}{(c)}
%1.18 \\
%1.41 \\
%1.76 \\
%\end{block}
%\end{blockarray},
%\vspace{-2em}
\end{equation}
And the achieved efficiency vector for users are denoted as $E=\langle1.18; 1.41; 1.76 \rangle$. 
The trading mechanism benefits all three users, leading to a much higher overall resource efficiency than the original max-min fairness scheduler. Gandiva$_{\mbox{fair}}$ is \textit{pareto-efficient}, however, it violates \textit{envy-freeness}, as the algorithm proceeds greedily, regardless of inter-user constraints. In this example, $u_3$ prefers $u_2$'s allocation. 

Another  limitation of Gandiva$_{\mbox{fair}}$ is that it cannot completely prevent all users from lying. Suppose $u_1$ artificially increases the speedup value on GPU$_2$ from 2 to 2.8. In this case, the price in the second-round trading changes from 2.5 to 2.9, so $u_1$ successfully wins more portions on GPU$_2$ and improves his efficiency, violating \textit{strategy-proofness}, as shown below:
\begin{equation*}
W^f =
\begin{blockarray}{ccc}
\mbox{GPU}_1 & \mbox{GPU}_2 & \\
\begin{block}{(cc)c}
1 & 2.8 & u_1 \\
1 & 3 & u_2 \\
1 & 4 & u_3 \\
\end{block}
&&
\end{blockarray},
X^f =
\begin{blockarray}{ccc}
\mbox{GPU}_1 & \mbox{GPU}_2 & \\
\begin{block}{(cc)c}
1 & 0.11 & u_1 \\
0 & 0.45 & u_2 \\
0 & 0.44 & u_3 \\
\end{block}
&&
\end{blockarray}.
%E^f = \begin{blockarray}{c}
%\begin{block}{(c)}
%1.22 \\
%1.35 \\
%1.76 \\
%\end{block}
%\end{blockarray}.
%\vspace{-2em}
\end{equation*}
In this case, the achieved efficiency vector becomes $E^f=\langle 1.22; 1.35; 1.76 \rangle$. Moreover, one can simply construct a more resource-efficient allocation that satisfies both \textit{envy-freeness} and \textit{sharing-incentive} properties, i.e.:
\begin{equation}
\label{optimal-allo-sharing}
X^{*} =
\begin{blockarray}{ccc}
\mbox{GPU}_1 & \mbox{GPU}_2 & \\
\begin{block}{(cc)c}
1 & 0 & u_1 \\
0 & 0.5 & u_2 \\
0 & 0.5 & u_3 \\
\end{block}
&&
\end{blockarray},
E^{*} = \begin{blockarray}{c}
\begin{block}{(c)}
1 \\
1.5 \\
2 \\
\end{block}
\end{blockarray}.
%\vspace{-2em}
\end{equation}
Without cheating, this allocation achieves much higher resource efficiency than Gandiva$_{\mbox{fair}}$, indicating Gandiva$_{\mbox{fair}}$ fails to provide optimal resource efficiency under \textit{sharing-incentive} constraints.

\noindent\textbf{Gavel} proposes a completely different approach by trying to speed up the least incentivized users in the cluster. In this sense, the ratio of efficiency obtained from Gavel to that under a max-min fair share is equalized across users. This policy is naturally \textit{sharing-incentive}. However, Gavel falls short to provide high resource efficiency and \textit{envy-freeness}. To be specific, consider the same speedup matrix as in Expression~\eqref{efficieny_allo_example}, the allocation and the related efficiency under Gavel is given by:
\begin{equation}
X =
\begin{blockarray}{ccc}
\mbox{GPU}_1 & \mbox{GPU}_2 & \\
\begin{block}{(cc)c}
0.91 & 0.09 & u_1 \\
0.09 & 0.45 & u_2 \\
0 & 0.45 & u_3 \\
\end{block}
&&
\end{blockarray} , \ \
E = \begin{blockarray}{c}
\begin{block}{(c)}
1.09 \\
1.44 \\
1.8 \\
\end{block}
\end{blockarray}, \ \
%\vspace{-2em}
\end{equation}
where $u_3$ prefers $u_2$'s allocation and the total resource efficiency attained is even lower than that under Gandiva$_{\mbox{fair}}$. Similarly, users under Gavel also have incentive to cheat on their speedup rate. For example, when $u_1$ increase his speedup rate on GPU$_2$ to 2.5, his resulted share vector become $\langle 0.85, 0.15 \rangle$, leading to higher resource efficiency. In this sense, Gavel does not achieve \textit{strategy-proofness} either. Moreover, $u_2$ can trade with $u_1$ to obtain higher efficiency without hurting anyone, that is, $x'_1=\langle 1, 0.045 \rangle$ and $x'_2=\langle 0, 0.0495 \rangle$. This demonstrates the \textit{pareto-inefficiency} and non-optimal resource efficiency under Gavel.

\begin{table}[!tb]
%\vspace{-.5em}
\centering
\caption{Properties guaranteed by existing schedulers.}
\begin{tabular}{|c|c|c|c|c|c|}
\hline
 & $\mathsf{PE}$ & $\mathsf{EF}$ & $\mathsf{SI}$ & $\mathsf{SP}$ & \textit{optimal efficiency}\\
\hline
Gavel~\cite{gavel} &  \xmark & \xmark & \cmark & \xmark & \xmark \\
\hline
Gandiva$_{\mbox{fair}}$~\cite{Gandivafair}  & \cmark & \xmark & \cmark & \xmark & \xmark \\
\hline
OEF & \cmark & \cmark & \cmark & \cmark  & \cmark \\
\hline
\end{tabular}
\label{tab:propertiestable}
%\vspace{-1em}
\end{table}

Upon examining existing schedulers, we find that, in addition to not achieving optimal efficiency, they fail to ensure various critical fairness properties outlined in Tab~\ref{tab:propertiestable}. One key reason is the tight coupling of these schedulers with a single fairness objective, such as Gavel and Gandiva$_{\mbox{fair}}$, aim to make their scheme sharing-incentive. These schedulers then try to enhance resource efficiency on a best-effort basis by introducing a second-stage optimization. However, due to the absence of global coordination between resource efficiency and fairness within a unified framework, they are far from effective. Additionally, as these schedulers assume jobs of each user can be represented by a speedup vector, they are not compatible with cases where multiple types of jobs are simultaneously trained by the same user, which consequently reduces their practicality.

\section{OEF Overview}
\subsection{Design Challenges}
\label{sec:challenge}

\subsubsection{Conflicts between efficiency and fairness}
\label{fair-efficiency}
The first challenges arise from the inherent conflicts between maximizing resource efficiency and maintaining fairness. 

% \begin{itemize}
{\small$\bullet$} \textit{Merely maximizing resource efficiency is detrimental.}
% \end{itemize}
Maximizing the overall training throughput within the cluster can be formulated as the following objective:
%One crucial goal is to effectively leverage heterogeneous resources to accelerate the training process of all DL workloads. This entails maximizing the overall training throughput within the cluster, which can be formulated as the following objective:
\begin{equation}
\max \sum_{j=1}^k \sum_{l=1}^n w_l^j \cdot x_l^j, \quad \mbox{s.t.,} \sum_{l=1}^n x_l^j \leq m_j, \ \mbox{for all} \ j.
\end{equation}
However, optimizing this objective alone yields an outcome that consistently assigns the highest scheduling priority to jobs with the greatest speedup on a fast GPU, which is inherently unfair to users with slower speedups. As an example, consider three users ($u_1,u_2,u_3$) with different speedups sharing a cluster of two different GPUs as follows:
\begin{equation}
%\label{efficieny_allo_example}
W =
\begin{blockarray}{ccc}
\mbox{GPU}_1 & \mbox{GPU}_2 & \\
\begin{block}{(cc)c}
1 & 2 & u_1 \\
1 & 3 & u_2 \\
1 & 4 & u_3 \\
\end{block}
&&
\end{blockarray}, 
X =
\begin{blockarray}{ccc}
\mbox{GPU}_1 & \mbox{GPU}_2 & \\
\begin{block}{(cc)c}
1 & 0 & u_1 \\
0 & 0 & u_2 \\
0 & 1 & u_3 \\
\end{block}
&&
\end{blockarray}.
%\vspace{-2em}
\end{equation}
It is straightforward to show that the best allocation result that produces the highest overall efficiency is to assign GPU$_2$ to $u_3$ and assign GPU$_1$ to $u_1$ or $u_2$. However, this allocation totally neglects the resource demand of $u_2$ (or $u_1$), which is not \textit{sharing-incentive} since $u_2$ does not get $\frac{1}{3}$ of all resources in the cluster. This allocation is also not \textit{envy-free} as both $u_1$ and $u_2$ prefer $u_3$'s allocation. Moreover, such scheme encourages $u_1$ and $u_2$ to lie, as jobs with greater speedup would be more likely to be placed on fast GPUs. Therefore, such scheme does not satisfy \textit{strategy-proofness}. In this sense, purely maximizing resource efficiency is completely unfair to users with slow speedup. Additionally, since users have incentive to lie on their speedup profiles, the resource allocation process will eventually behave like max-min fairness, thereby greatly degrading the overall resource efficiency that can be achieved in the cluster.

{\small$\bullet$} \textit{Naively maintaining fairness properties reduces resource efficiency}. Placing excessive emphasis on fairness attributes such as \textit{envy-freeness} can actually hinder resource efficiency, leading to a situation akin to the classic prisoner's dilemma problem~\cite{hug}. To demonstrate this, let's consider an example where two users share a cluster comprising two distinct GPUs, characterized by a speedup matrix represented by the following matrix, labeled as $W$. When attempting to optimize resource efficiency and maintain \textit{envy-freeness} simultaneously, it results in the allocation $X$ as follows:
\begin{equation}
\label{speedup-matrix-example}
W =
\begin{blockarray}{ccc}
\mbox{GPU}_1 & \mbox{GPU}_2 & \\
\begin{block}{(cc)c}
1 & 2 & u_1 \\
1 & 5 & u_2 \\
\end{block}
&&
\end{blockarray},
X =
\begin{blockarray}{ccc}
\mbox{GPU}_1 & \mbox{GPU}_2 & \\
\begin{block}{(cc)c}
1 & 0.25 & u_1 \\
0 & 0.75 & u_2 \\
\end{block}
&&
\end{blockarray},
%\vspace{-2em}
\end{equation}
by following the envy-freeness constraints for $u_1$ and $u_2$ listed as below:
\begin{equation*}
x_1^1 + 2 x_1^2 \geq x_2^1 + 2 x_2^2, \ \mbox{and} \quad x_2^1 + 5 x_2^2 \geq x_1^1 + 5 x_1^2.
\end{equation*}
The above allocation $X$ yields a total resource efficiency of 5.25. In contrast, when the first user cheats the cluster manager and lies about his speedup on GPU$_2$ up to 4, the resulted allocation will become $X^{l}$ listed below, which leads to a total throughput of 1.75 for $u_1$ and is 16.7\% higher than his original throughput. However, the overall efficiency will decrease from 5.25 to 4.875, implying that naively preserving \textit{envy-freeness} can do more harm than good. This example further suggests that after $u_1$ lies, another user $u_2$ also has the incentive to lie to increase his throughput. Finally, the result $X^{f}$ will be equivalent to the allocation under max-min fairness as shown below, which results in no better throughput for both users, that is, with 
$u_1$ obtaining 1.5 and $u_2$ achieving 3.
\begin{equation*}
\label{false-si}
X^{l} =
\begin{blockarray}{ccc}
\mbox{GPU}_1 & \mbox{GPU}_2 & \\
\begin{block}{(cc)c}
1 & 0.375 & u_1 \\
0 & 0.625 & u_2 \\
\end{block}
&&
\end{blockarray},
X^{f} =
\begin{blockarray}{ccc}
\mbox{GPU}_1 & \mbox{GPU}_2 & \\
\begin{block}{(cc)c}
0.5 & 0.5 & u_1 \\
0.5 & 0.5 & u_2 \\
\end{block}
&&
\end{blockarray},
%\vspace{-2em}
\end{equation*}

Under the same speedup matrix, \textit{sharing-incentive} requires the allocation to matrix satisfy:
%\vspace{-.3em}
\begin{equation*}
x_1^1 + 2 x_1^2 \geq 1.5, \ \mbox{and} \quad x_2^1 + 4 x_2^2 \geq 3.
%\vspace{-.3em}
\end{equation*}
% As such, optimizing resource efficiency gives the same allocation as above. 
When $u_1$ lies about his speedup on GPU$_2$ to 4, his own efficiency improves whereas the overall resource efficiency drops by 8\%.

% The main reason for this inefficiency is the absence of \textit{strategy-proofness}. In other words, \textit{strategy-proofness} turns out to be the key to achieving better resource efficiency for heterogeneous GPU scheduling.

\subsubsection{Conflicts between various fairness properties}
\label{conflict_fairness}
The second challenge stems from the inherent conflicts between \textit{strategy-proofness} and other fairness properties when striving to maximize resource efficiency in a heterogeneous context. These conflicts are illuminated through the following theorems. This highlights that it is not feasible to provide users with both properties at the same time. Therefore,
we need to carefully integrate fairness properties with resource efficiency when designing resource
allocation schemes.
% , which establish the foundation for the selection of diverse properties that harmonize with resource efficiency within the OEF framework.
% as demonstrated in the following theorems. These theorems establish a foundation for OEF to select various properties to align with resource efficiency.

\begin{lemma}
\label{theo:lemma1}
% All the non-zero elements in $x_1$ (user with the slowest speedup) lie on the leftmost positions and of value $m_j$ except for the right-most one in pursuit of optimal efficiency.
To achieve optimal efficiency, in $\vec{x_1}$, which represents the allocation for the user with the slowest speedup, all non-zero elements are positioned at the leftmost positions and have a value of $m_j$, except for the right-most element. 
\end{lemma}
\begin{proof}
In order to maximize the overall efficiency, the lowest-end GPU quotas must be allocated user 1 who possesses the lowest speedup among all users. Referring to the definition of speedup in \cref{fairnees-definition}, it is evident that $w^j_1=\min_l w^j_l,\forall j\in [1,\cdots,k]$, and $w^p_1 < w^q_1, \ \forall 1\leq p < q \leq k$. These inequalities imply that the non-zero elements in $\vec{x_1}$ are filled from left to right, and $x^q_1 > 0$ holds only when $x^p_1=m_p, \forall p<q$.
\end{proof}

\begin{theorem}
\label{theo:ef_violation}
No sharing mechanism providing optimal resource efficiency can simultaneously guarantee \textit{envy-freeness} and \textit{strategy-proofness}.
\end{theorem}
\begin{proof}
To achieve optimal efficiency, we should have stringent \textit{envy-free} constraints on $\vec{x_1}$ and another user $\vec{x_l}$ such that:
\begin{equation}
\label{envy-conflict-strategy}
\vec{W_1}\cdot \vec{x_1} - \vec{W_1}\cdot \vec{x_l} = 0,\ l\neq1, 
\end{equation}
otherwise, allocating a larger GPU share to users with higher speedup values could enhance the overall resource efficiency. Let us examine a scenario where user 1 manipulates and marginally increases the speedup vector, while ensuring it remains below the speedup of the second slowest user, i.e., 
\begin{equation*}
\vec{W_2}\succcurlyeq \vec{W'_1} = \langle 1,w^2_1+\epsilon^2_1,\cdots, w^k_1+\epsilon^k_1 \rangle\succcurlyeq \vec{W_1}, \ \mbox{and} \ \epsilon^j_1 > 0, \forall j.
\label{eq:fake_speedup}
\end{equation*}
If the rightmost non-zero element $x^j_1$ does not reach the limit $m_j$, it becomes apparent that the inequality $x^{j'}_1> x^j_1$ holds by utilizing Lemma~\ref{theo:lemma1}. Consequently, we have $\vec{W_1}\cdot\vec{x'_1} > \vec{W_1}\cdot\vec{x_1}$. In another scenario, it is evident that the leftmost zero element $x^{(j+1)'}_1$ improves. Both of these results contradict \textit{strategy-proofness}. This completes the proof.
\end{proof}

\begin{figure}[!tb]
\begin{minipage}[!htb]{0.44\linewidth}
\centering
\includegraphics[width=0.99\linewidth]{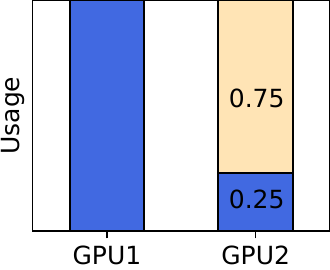}
\subcaption{Before lying}
\end{minipage}
\begin{minipage}[!htb]{0.44\linewidth}
\centering
\includegraphics[width=0.99\linewidth]{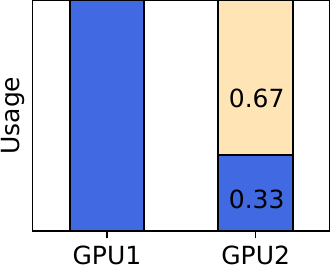}
\subcaption{After lying}
\end{minipage}
%\vspace{-.5em}
\caption{An example shows the conflicts between various fairness properties. The blue and yellow space represent the allocation for the first and second user, respectively.}
%\vspace{-1.5em}
\label{fig:ef_example}
\end{figure}

\begin{theorem}
\label{theo:si_violation}
No sharing mechanism providing optimal resource efficiency can guarantee \textit{sharing-incentive} and \textit{strategy-proofness} at the same time.
\end{theorem}
\begin{proof}
In this scenario, user 1 is also subjected to a stringent \textit{sharing-incentive} constraint on the allocation vector $\vec{x_1}$:
%\vspace{-.5em}
\begin{equation}
\label{si-conflict-strategy}
\vec{W_1}\cdot \vec{x_1} - \vec{W_1}\cdot\vec{\frac{m}{n}} = 0,\ \vec{m}=\langle m_1,\cdots,m_k\rangle,
%\vspace{-.5em}
\end{equation}
We construct a fake speedup vector, $\vec{W'_1} = \langle 1, w^2_1, \cdots, w^k_1 + \epsilon^k_1 \rangle$, where only the last value is augmented. Under the real speedup, we assert that $x^k_1 <\frac{m_k}{n}$, as having $x^k_1\geq\frac{m_k}{n}$ would lead to $x^j_1=m_j, \forall j\neq k $. This condition would imply $\vec{W_1}\cdot \vec{x_1} > \vec{W_1}\cdot\vec{\frac{m}{n}}$, which contradicts Equation~\eqref{si-conflict-strategy}. Therefore, we can conclude that  $\vec{W_1}\cdot \vec{x_1} \leq \vec{W'_1}\cdot \vec{x_1} < \vec{W'_1}\cdot \vec{\frac{m}{n}}$. This implies that $\vec{x'_1}\succcurlyeq \vec{x_1}$ by utilizing Lemma~\ref{theo:lemma1}, violating \textit{strategy-proofness}. This completes the proof.
\end{proof}

Moreover, we provide an intuitive example depicted in Fig.~\ref{fig:ef_example} to better explain these two theorems.  For the first theorem, consider the scenario where two users with a speedup matrix of $\langle 1,2;1,4\rangle$. In this case, the allocation that provides the optimal efficiency and \textit{envy-freeness} is $\langle 1,0.25;0,0.75\rangle$\footnote{As the second user has the highest speedup rate, giving more high-end GPU devices to her benefits the overall efficiency most. However, to maintain \textit{envy freeness}, the user with speedup $\langle1,2\rangle$ should not envy the share of another user barely.}. When the first user lies and reports the speedup to be $\langle 1,3\rangle$, the original allocation is no longer \textit{envy-free}\footnote{Since $\langle 1,3\rangle\cdot\langle 1,0.25\rangle < \langle 1,3\rangle\cdot\langle 0,0.75\rangle$, the user with false speedup will envy the share of another one.}, and the new allocation with the same goal becomes $\langle 1,0.33;0,0.67\rangle$, thus violating \textit{strategy-proofness}.

\subsection{Architecture \& Workflow}
In light of the conflicts that arise from different fairness properties, our primary objective is to develop allocation schemes for diverse scheduling scenarios, with the aim of improving resource efficiency while maintaining fairness properties to the greatest extent possible. Specifically, OEF prioritizes \textit{strategy-proofness} while maximizing resource efficiency in non-cooperative environments where multiple tenants may cheat to compete for limited resources. In cooperative environments where cheating is a non-issue, it ensures \textit{envy-freeness} and \textit{sharing-incentive} while maximizing the resource efficiency. Distinct scheduling mechanisms are adopted based on the selection of the cluster administrator.  Moreover, the \textit{pareto-efficiency} are always maintained in both environments. The complete system architecture of OEF, along with its scheduling workflow, is depicted in Fig.~\ref{fig:sysarch}.

Once a tenant enters the cluster, OEF takes charge of the submitted DL jobs. First, OEF provides tenants with an profiling agent, to which users submit one of their DL training tasks for throughput estimation. Subsequently, the agent delivers the profiled throguhput to the scheduler module, which is responsible for generating optimal resource allocation schemes. At its core, a fair share evaluator (\ding{182} in Fig.~\ref{fig:sysarch}) first computes the fair share among users based on the algorithm specifically tailored for a selected scheduling environment.  Importantly, the preservation of \textit{strategy-proofness} in the non-cooperative environment is guaranteed through the enforcement of OEF's allocation mechanism. The placer component (\ding{183} in Fig.\ref{fig:sysarch}) then allocates GPUs to users, enforcing efficient allocation in the long run. Specifically, we design a rounding based placement policy to approximate the fractional ideal GPU shares among tenants in multiple rounds, and address the straggler effect incurred by cross-type placement of DL training jobs.

\begin{figure}
\centering
\includegraphics[width=0.95\linewidth]{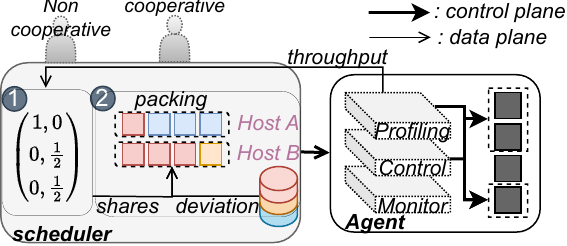}
%\vspace{-.5em}
\caption{The system architecture of OEF.}
%\vspace{-1.5em}
\label{fig:sysarch}
\end{figure}

\section{OEF Design}
\subsection{Profiling}
In OEF, tenants have the flexibility to submit multiple DL training jobs of the same type or diverse types. To facilitate this, the profiling agent offers tenants a profiling interface, allowing them to collect throughput information for their DL training jobs on each type of GPU. Specifically, tenants are responsible for selecting a representative profiling task for each job type they have. The profiling mechanism is lightweight, incurring minimal profiling overhead by executing a short run (few mini-batches in our setting). This makes sense since the execution time of each iteration almost keeps constant over time. Upon the completion of the profiling process, the final speedup vector is generated and provided to OEF scheduler, which will then makes allocation decisions to each tenant. 

% \textcolor{red}{
% Remarkably, though precise throughput information is critical for providing high efficiency and maintaining \textit{strategy proofness}, OEF does not implement any anti-fraud mechanism such as code inspection in the profiling process. Conversely, we achieve the \textit{strategy-proofness} via introducing an allocation mechanism discussed in~\cref{noncooperative-mechanism}, under which tenants has no incentive to report inauthentic throughput information.
% }

%This profiling module also monitors the new jobs submitted by each user and detects differences in speedup rate. If the new job has a new speedup rate, the virtual user splitting process is triggered and new shares of virtual users are reallocated to maintain fairness and efficiency.

\subsection{Allocation Mechanisms}
\subsubsection{OEF in non-cooperative environments}
\label{noncooperative-mechanism}
In contrast to the requirements of \textit{envy-freeness} or \textit{sharing-incentive}, expressing \textit{strategy-proofness} directly through a formulation is challenging. This difficulty arises from the fact that if a user misrepresents the speedup vector, the speedup matrix $W$ and allocation matrix $X$ change accordingly, making it difficult to compare the efficiency between before and after cheating among users. Fortunately, we have discovered that ensuring all users achieve identical (normalized) throughput can guarantee \textit{strategy-proofness}. We formally design the allocation scheme based on the optimal solution to the following optimization problem:
\begin{subequations}
\begin{align}
\max & \sum_{l=1}^n \sum_{j=1}^k w_l^j \cdot x_l^j \label{eq:objectiveSP}\\
   \mbox{s.t.}, & \sum_{l=1}^n x_l^j \leq m_j,\quad \forall j,\label{eq:PEconstraintSP}\\
   & \sum_{j=1}^k w_l^j \cdot x_l^j = \sum_{j=1}^k w_i^j \cdot x_i^j, \ \forall 1 \leq i, l \leq n. \label{eq:eqinnonco}
\end{align}
\label{eq:nonco}
\end{subequations}

The key idea behind this optimization is that attempting to cheat will lead to a new allocation that still maintains the same overall speedup across all users. If the efficiency of honest users decreases under the new allocation, it follows that the efficiency of the cunning user under their actual speedup should also decrease due to the equality constraint~\eqref{eq:eqinnonco}. Conversely, if honest users receive more resources because of improved efficiency, it indicates that cunning users must be allocated fewer resources than they had received before.

\subsubsection{OEF in cooperative environments}
As \textit{strategy-proofness} is not a major concern in cooperative environments, we get rid of the conflicts identified in \cref{conflict_fairness} and are able to incorporate \textit{envy-freeness} and \textit{sharing-incentive}  into a unified framework. To be more precise, we formulate the cooperative OEF framework as follows:
%\begin{equation*}
%    \omega_j\cdot x_j \geq \omega_j\cdot\vec{\frac{1}{n}},\quad\forall j\in\mathbb{U}.
%\end{equation*}
%Analogously, we can formulate Envy Freeness via enforcing inequality that efficiency under her own share should always be higher than that under the share of others:
%\begin{equation*}
%    \omega_j\cdot x_j \geq \omega_j\cdot x_k,\quad\forall j,k\in\mathbb{U}, j\neq k.
%\end{equation*}
%\vspace{-.5em}
\begin{subequations}
\begin{align}
\max & \sum_{l=1}^n \sum_{j=1}^k w_l^j \cdot x_l^j \label{eq:objectivePESIEF}\\
   \mbox{s.t.}, & \sum_{l=1}^n x_l^j \leq m_j,\quad \forall j, \label{eq:PEconstraintPESIEF}\\
   & \sum_{j=1}^k w_l^j \cdot x_l^j \geq \sum_{j=1}^k w_l^j \cdot x_i^j, \ \forall 1 \leq i, l \leq n. \label{eq:EFinPESIEF}
\end{align}
\label{eq:PESIEF}
% \vspace{-.5em}
\end{subequations}
% \vspace{-.2em}

In this framework, we address the need for \textit{envy-freeness} by directly incorporating it as optimization constraints. However, we do not include explicit constraints for expressing \textit{sharing-incentive}. This is because, as we will demonstrate in Theorem~\ref{share-incentive-oef}, the pursuit of \textit{envy-freeness} inherently leads to \textit{sharing-incentive} when we maximize the overall resource efficiency.

\subsubsection{Weighted OEF}
Currently, administrators frequently assign varying levels of priority to users. Surprisingly, it is possible to extend OEF to accommodate this scenario by allocating additional resources to users engaged in more critical tasks or those with higher resource quotas~\cite{Gandivafair}. 

% In this section, we introduce Weighted OEF that is applicable to both cooperative and non-cooperative environments.

In the weighted OEF framework, each tenant $j$ is assigned a weight $\pi_j$ that represents their level of importance. Specifically, if $\pi_2=2$ and $\pi_1=1$, it implies that the second user $u_2$ should receive an allocation resulting in twice the total training throughput compared to the first user $u_1$. OEF introduces a novel approach to achieve this while ensuring fairness. The key concept involves replicating the speedup vector multiple times for users with higher weights. For instance, consider the speedup matrix in Eq.~\eqref{speedup-matrix-example} with $\pi_2=2$, OEF will modify the speedup matrix and produce an allocation in a non-cooperative environment as follows:

% Under weighted OEF, each tenant $j$ is associated with a weight $\pi_j$, which indicates his importance.
% In particular, if $\pi_2=2$ and $\pi_1=1$, the second user $u_2$ is expected to get an allocation that yields twice the total training throughput of the first user $u_1$. OEF introduces a novel method to handle this while guaranteeing multiple fairness properties. The main idea is to replicate the speedup vector multiple times for users with higher weight. For example, for the speedup vector given in \cref{fair-efficiency} with $\pi_2=2$, OEF will yield a new speedup matrix and produce an allocation in non-cooperative environment as:
%\vspace{-1.4em}
% {\color{red}
%\vspace{-1em}
\begin{equation*}
\label{efficieny_allo}
W =
\begin{blockarray}{ccc}
\mbox{GPU}_1 & \mbox{GPU}_2 & \\
\begin{block}{(cc)c}
1 & 2 & u_1 \\
1 & 5 & u_2 \\
1 & 5 & u_2 \\
\end{block}
&&
\end{blockarray}, \ \
X =
\begin{blockarray}{ccc}
\mbox{GPU}_1 & \mbox{GPU}_2 & \\
\begin{block}{(cc)c}
1 & 1/3 & u_1 \\
0 & 1/3 & u_2 \\
0 & 1/3 & u_2 \\
\end{block}
&&
\end{blockarray}.
%\vspace{-2em}
\end{equation*}
% }
In this scenario, if there is only one worker available for two types of GPUs, $u_2$ will be allocated ${2}/{3}$ fraction of the second type of GPU. This replication approach employed by Weighted OEF ensures that all the fairness properties achieved by OEF are perfectly maintained. More importantly, it enables a user to run multiple DL jobs with varying speedups, as explained in \cref{comptibility}. However, conventional solutions that directly incorporate user weights into optimization problems are unable to handle such cases.
% In this case, $u_2$ will be allocated a fraction of the second type of GPU of ${2}/{3}$, if only one worker is available for two types. This replication approach enables Weighted OEF to perfectly maintain all the fairness properties achieved by OEF. Furthermore, it allows a user to run multiple DL jobs of different speedups, as explained in \cref{comptibility}. However, conventional solutions that directly express user weights in optimization problems cannot handle such cases.
% This weight value denotes how many unit of fair share this user should obtain under OEF, i.e., the number of copies of this user involved in the speedup matrix. Exemplified by a simple example, $speedup=\begin{pmatrix} 1 & 2 \\ 1 & 3 \end{pmatrix}$, where the second user is two-fold important than the first user. In this case, OEF will copy a vector of the second user and stack it into the speedup matrix, i.e., $speedup'=\begin{pmatrix} 1 & 2 \\ 1 & 3 \\ 1 & 3 \end{pmatrix}$. In this case, the resulting share values from the second and the third rows all belong to the user of weight 2. Specifically,this scheme offers the user with two-fold weight a throughput twice higher than that of the rival in non-cooperative setting, and provides the twice-important user with double shares, where each is with Envy-Free and Sharing-Incentive. More importantly, if the weights of all users are set to the same value, this weighted version trivially reduces to non-weighted ERF.

\subsubsection{Compatible with different job types}
\label{comptibility}
% So far, we assume that all the jobs submitted by the same tenant have the same speedup vector.
% % In the above section, we discuss how to efficiently allocate resources to users, based on the assumption that jobs submitted from each user are of the same type, indicating the same speedup rate in each GPU type.
% However, this assumption is not flexible enough to cope with all real scenarios. 
In certain situations, a user may have the need to train multiple different types of DL jobs concurrently. However, representing these distinct jobs within a single speedup vector can be challenging. Fortunately, the framework of Weighted OEF allows us to effectively address this scenario.

When a user submits different types of DL jobs to a cluster, OEF treats each job type as a separate virtual user. This allows us to handle the different job types individually. However, it is worth noting that simply creating virtual users may result in an unfair increase in a user's overall throughput. To address this, we divide the weight of a user equally among their job types and assign a corresponding weight to each virtual user.

Once again, we will use the speedup matrix in Eq.~\eqref{speedup-matrix-example}, as an example to demonstrate the allocation process. Consider that both users have equal weights and the first user wishes to train a new DL job with a speedup vector of 
$\langle 1,3\rangle$ in addition to the original job, the resulting new speedup matrix and the allocation in a non-cooperative environment are as follows:

% Again, we take the speedup matrix of two users on two different types of GPUs described in \cref{fair-efficiency} as an example to illustrate the allocation process.  Suppose these two users have the same weight and the first user wants to train a new DL job with a speedup vector $\langle 1,3\rangle$ in addition to the original job, the new speedup matrix and the allocation in a non-cooperative environment are given by:
%\vspace{-1em}
% \color{red}
%\vspace{-1em}
\begin{equation*}
\label{efficieny_allo}
W =
\begin{blockarray}{ccc}
\mbox{GPU}_1 & \mbox{GPU}_2 & \\
\begin{block}{(cc)c}
1 & 2 & u_1 \\
1 & 3 & u_1 \\
1 & 5 & u_2 \\
1 & 5 & u_2 \\
\end{block}
&&
\end{blockarray}, \ \
X =
\begin{blockarray}{ccc}
\mbox{GPU}_1 & \mbox{GPU}_2 & \\
\begin{block}{(cc)c}
1 & 0.11 & u_1 \\
0 & 0.41 & u_1 \\
0 & 0.24 & u_2 \\
0 & 0.24 & u_2 \\
\end{block}
&&
\end{blockarray}.
%\vspace{-2.5em}
\end{equation*}
By this design, OEF is also able to support different priorities of jobs within the same tenant.

% a setting user 1 and user 2 share a cluster and forms a speedup matrix $speedup=\begin{pmatrix} 1 & 2 \\ 1 & 3 \end{pmatrix}$, and the user 1 want to train another collections of jobs with speedup vector $(1,4)$. As such, two virtual users are created and each are provided with a weight $\frac{1}{2}$. In this case, the user 2 is twice important than two virtual users, doubling her speedup vector. Therefore, the new speedup matrix becomes $speedup'=\begin{pmatrix} 1 & 2 \\ 1 & 3 \\ 1 & 3 \\ 1 & 4 \end{pmatrix}$, where the resulting GPU share of the first and the fourth rows belong to user 1, and that from the second and the third rows belong the user 2.

\subsection{Placement Optimization}
\label{sec:alloc-mechanism}
 % \textcolor{red}{It is worth noting that OEF algorithms offer fractional share values to tenants, which preplexs the allocation process.}
% Given the significant interference caused by multiplexing, it is more beneficial for the OEF system to allocate the entire GPU to a single user in each scheduling round. However, it is possible that OEF offers fractional share values to tenants. To address this, we have devised a rounding policy that converts the fractional value assigned to each user into an integral share, maintaining fairness and efficiency in the long run. 
Given the substantial interference caused by multiplexing~\cite{xiao2020antman}, it proves more advantageous for the placer in OEF to allocate the entire GPU to a single user in each scheduling round. However, there is a possibility that the fair share evaluator may yield fractional share values. To mitigate this, we have developed a rounding policy that converts the fractional value assigned to each user into an integral share, ensuring fairness and efficiency in the long run. During each round $t$, the placer tracks the cumulative deviation $\mathsf{dev}_j(t)$ for every user $j$ from their ideal share $\mathsf{ideal}_j(t)$. Utilizing this deviation, OEF assigns each user $j$ a share, represented as $\mathsf{real}_j(t) = \mathsf{round}\big(\mathsf{ideal}_j(t) + \mathsf{dev}_j(t)\big)$. The deviation is updated based on the formula $\mathsf{dev}_j(t+1) = \mathsf{dev}_j(t) + \big(\mathsf{ideal}_j(t) - \mathsf{real}_j(t)\big)$, allowing the actual share for each job to to gradually approximate the assigned value over time.

In situations where a user is allocated only a small GPU share, insufficient to meet the requirements of any job, we have refined the rounding strategy. Specifically, the placer adjusts $\mathsf{real}_j(t)$ to $0$ if it falls below $\min_k \mathsf{demand}_k$, where $\mathsf{demand}_k$ represents the required worker size of job $k$ under user $j$. By progressively accumulating the deviation value over time, tenants who experience resource starvation are assured a $\mathsf{real}_j(t)$ greater than $\min_k \mathsf{demand}_k$. Consequently, these tenants will inevitably have the opportunity to execute at least one job.% In this sense, a user may be less frequent to be allocated resources, but must have enough workers to schedule their own jobs in their turn.

To further alleviate network contention within the cluster, OEF optimizes the placement scheme for DL workloads. Notably, the placer grants placement priority to jobs with more workers, as the collective communication overhead— a primary contributor to network congestion—escalates with the number of workers. Consequently, under OEF, these jobs are prioritized for placement on a host whenever feasible.
%The primary focus is on efficiently packing larger jobs that require multiple workers alongside single-worker jobs on each host.

% \vspace{-.3em}
\subsection{Alleviating the Straggler Effect}
\label{distributed-oef}
It is possible for users to possess multiple types of GPUs, requiring the distribution of their jobs across different GPU types under OEF. In such cases, the occurrence of the \textit{straggler effect} is a concern, where the slowest GPU type limits the overall speed of DL training, outweighing the advantages provided by faster GPUs. However, OEF's optimization ensures that the allocated GPUs for each user are limited to a small range, as proven in Theorem \ref{adjacent-oef}. This advantage greatly mitigates the impact of the straggler effect.
% Specifically, OEF employs a strategy similar to that in \cite{semidynamic}, where it tests various batch size combinations by gradually increasing the batch sizes assigned to high-end GPUs. This approach aims to reduce the training speed gap between different GPU types as much as possible.

Furthermore, we provide a theoretical finding that, with OEF, only a few users and their respective jobs will experience the straggler effect. According to the Extreme Point Theorem~\cite{extreme-convex}, the allocation matrix produced by OEF contains a maximum of $(n+m-1)$ nonzero elements. Hence, when the number of GPU types is smaller than the number of tenants, most tenants will be allocated only one type of GPU. This observation highlights the effectiveness of OEF in minimizing the impact of the straggler effect.

\subsection{Implementation Details}
We have developed a prototype system for OEF using Python 3. The current implementation supports DL training jobs written in PyTorch.

\noindent\textbf{Fair share evaluation.} Given the linearity inherent in our algorithm formulation, we have implemented the fair share evaluator in OEF utilizing \textit{cvxpy}, an effective tool for solving convex optimization problems. In this context, resource allocation is defined as a \textit{variable}, and fairness constraints are articulated as \textit{constraints}. We employ the \textit{ECOS} solver to efficiently generate the allocation.

\noindent\textbf{Control and data plane.} 
Within OEF, inter-module communication is established through Restful interfaces. Additionally, the seamless transfer of checkpoint files between hosts is facilitated by rsync, ensuring an imperceptible training process for the user.

\section{Properties of OEF}
\label{property-oef}
In this section, we present the main properties achieved by the OEF framework. 
% Due to space limit, we include the omitted proofs in the technical report~\cite{technical-report}.
\begin{enumerate}[leftmargin=*]
\item
\noindent OEF satisfies \textit{pareto-efficiency} (\textbf{Theorem} \ref{optimal-efficiency}).
\item
\noindent In non-cooperative environments, OEF ensures \textit{strategy-proofness} (\textbf{Theorem} \ref{strategy-proof-oef}). To our knowledge, OEF is the first resource allocation framework that can achieve \textit{strategy-proofness} for heterogeneous GPU clusters.
\item    In cooperative environments, OEF achieves the optimal resource efficiency and ensures \textit{sharing-incentive} and \textit{envy-freeness} (\textbf{Theorem} \ref{share-incentive-oef}).
\item OEF only assigns GPUs of adjacent types to the same user (\textbf{Theorem} \ref{adjacent-oef}). OEF can benefit from this property to better support alleviating the potential straggler effect in distributed DL training (\cref{distributed-oef}).
\end{enumerate}

\begin{theorem}
\label{share-incentive-oef}
Cooperative OEF achieves the optimal resource efficiency, while ensuring \textit{sharing-incentive} and \textit{envy-freeness}.
\end{theorem}
%\vspace{-.5em}
\begin{proof}
As discussed above, envy-freeness denotes that the GPU share assigned to each user is already the best-fit one in contrast to that of others, i.e., $\sum_{j=1}^k w_l^j \cdot x_l^j \geq \sum_{j=1}^k w_l^j \cdot x_i^j, \ \forall 1 \leq i, l \leq n$. Summing up all constraints for a user $l$, including a constraint that user $l$ should not envy the share he has already obtained, we have:
%\vspace{-.5em}
\begin{equation*}
n\cdot \sum_{j=1}^k w_l^j \cdot x_l^j \geq \sum_{i=1}^n \sum_{j=1}^k w_i^j \cdot x_i^j.
%\vspace{-.5em}
\end{equation*}

Note that we also have $\sum_{i=1}^n x^j_i = m_j$ for each type $j$ of GPU resources. In this sense, by moving $n$ to the left side, we can obtain:
\begin{equation*}
\sum_{j=1}^k w_l^j \cdot x_l^j \geq \sum_{j=1}^k \frac{m_j}{n}\cdot w_l^j,
\end{equation*}
which directly implies \textit{sharing-incentive}. As in cooperative environments OEF explicitly formulates \textit{envy-freeness} constraint, these two properties are automatically achieved. Furthermore, the OEF framework guarantees the best resource efficiency, since linear optimization yields a globally optimal solution. This completes the proof.
\end{proof}

\begin{theorem}
\label{adjacent-oef}
Under OEF, users are allocated adjacent GPU types.
\end{theorem}
\begin{proof}
In this part, we prove that OEF can allocate adjacent GPU types only in cooperative and non-cooperative environments respectively. 

First, in the non-cooperative setting, we first demonstrate that OEF does not assign a non-zero value to the lower left corner of another non-zero value. To achieve this, we construct a counter-example as follows:
\begin{equation*}
X =
\begin{pmatrix} x_1 & \cdots & y_1 \\ \vdots & \ddots & \vdots \\ x_2 & \cdots & y_2 \end{pmatrix},\quad
W =
\begin{pmatrix} a_1 & \cdots & a_1\cdot b_1 \\ \vdots & \ddots & \vdots \\ a_2 & \cdots & a_2\cdot b_2 \end{pmatrix},
\end{equation*}
where $y_1$ and $x_2$ are non-zero, $a_2 > a_1$, and $b_2 > b_1$. Note that two users should have the same efficiency according to the formulation of non-cooperative OEF. However, if we modify the allocation matrix $X$ to be as follows:
\begin{equation*}
\begin{pmatrix} x_1 + x_2 & \cdots & y_1 - \epsilon_1 \\ \vdots & \ddots & \vdots \\ 0 & \cdots & y_2 + \epsilon_2 \end{pmatrix}, \quad or \quad
\begin{pmatrix} x_1 + \epsilon_3 & \cdots & 0 \\ \vdots & \ddots & \vdots \\ x_2 - \epsilon_4 & \cdots & y_1 + y_2 \end{pmatrix},
\end{equation*}
we can keep the efficiency of each user unchanged while leaving spare resources in $x$ or $y$ column, which can be utilized to improve the efficiency of all users to the same extent. To be specific, in the first case, we have $a_1 x_1 + a_1 b_1 y_1 = a_1 (x_1 + x_2) + (y_1 - \epsilon_1) a_1 b_1$ and $a_2 x_2 + a_2 b_2 y_2 = a_2 b_2 (y_2 + \epsilon_2)$. Spare resources exist because  $\epsilon_1 -  \epsilon_2 = \frac{x_2}{b_1} - \frac{x_2}{b_2} > 0$. In this sense, maximum resource efficiency is not achieved under allocation $X$, violating the optimality of OEF. This indicates that it is impossible for OEF to produce allocations like $X$.
As such, an allocation that has a zero value between non-zero values with the following pattern:
%\vspace{-1.2em}
\begin{equation*}
  X =
\begin{pmatrix}
   &        & x_1 &        &     \\
   &        & \vdots &        &     \\
x^1_2 & \cdots & 0   & \cdots & x^2_2 \\
   &        & \vdots &        &     \\
   &        & x_3 &        &     \\
\end{pmatrix},
%\vspace{-.6em}
\end{equation*}
will not exist, as either $x_1$ or $x_3$ is non-zero, it will locate on the lower left corner of another non-zero entry. Therefore, users are allocated adjacent GPU types only.

We proceed to prove the theorem under cooperative OEF. Note that in cooperative OEF, no inter-user efficiency equalization is enforced and each user is only required to have an efficiency greater than a certain value (for \textit{envy-freeness} constraints, this value is $\min_l w_i\cdot x_l, \forall l$ given a user $i$). In this sense, we can directly trade between the first and the second column, or trade between the second and the third column within $X$ to generate higher efficiency without hurting anyone, thus contradicting the optimality of OEF. This completes the proof.
\end{proof}

\begin{theorem}
\label{optimal-efficiency}
OEF achieves pareto-efficiency.
\end{theorem}
% \vspace{-.5em}
\begin{proof}
Assuming that the allocation produced by OEF is "pareto-inefficient", which means that it is possible to perform effective trading to improve overall efficiency without negatively impacting anyone. In this sense, there must exist another allocation, denoted as $X'$, that results in a greater overall efficiency than the OEF allocation $X$, while still satisfying the same constraints, i.e., $\sum_{j=1}^k \sum_{l=1}^n w^j_l\cdot x^{j'}_l > \sum_{j=1}^k \sum_{l=1}^n w^j_l\cdot x^j_l $. However, this contradicts the fact that $X$ is the globally optimal solution within the same feasible domain in linear optimization. This completes the proof.
\end{proof}

\begin{theorem}
\label{strategy-proof-oef}
Non-cooperative OEF ensures \textit{strategy proofness}.
\end{theorem}
% \vspace{-.3em}
\begin{proof}
We establish the proof by examining the new allocation, denoted as  $X'$, that occurs after someone cheats. If the efficiency of honest users decreases, the cunning user $i$ will also be penalized due to the constraint of efficiency equality between users, i.e., $\vec{w_i}\cdot\vec{x'_i} < \vec{w'_i}\cdot\vec{x'_i} = \vec{w_k}\cdot\vec{x'_k}$ where $k$ denotes the honest user. This inequality ensures the property of being \textit{strategy-proof}. In another scenario, where the efficiency of honest users increases, we only have to analyze the case where user $i$ increases a small amount of acceleration rate and still does not exceed any speedup rates of others. Without loss of generality, assume user $i$ increases his speedup rate from $\vec{w_i}$ to $\vec{w_i'}$, which satisfies $\vec{w}_{i+1} \succcurlyeq \vec{w_i'} \succcurlyeq \vec{w_{i}}$. Following Theorem \ref{adjacent-oef}, the allocation undergoes the following changes:
 % \vspace{-.5em}
\begin{equation*}
W=\begin{pmatrix}
    a_i & b_i\\
    a_{i+1} & b_{i+1}
\end{pmatrix} \
X=\begin{pmatrix}
    x_i & y_i\\
    0 & y_{i+1}
\end{pmatrix} \
X’=\begin{pmatrix}
    x_i' & y_i'\\
    0 & y_{i+1}'
\end{pmatrix},
\end{equation*}
where the bottom-left zero value in $X$ will remain unchanged in $X'$, otherwise the optimal resource efficiency cannot be attained. Although there are various possible structures for the allocation matrix before and after cheating, we can employ the same method to demonstrate the satisfaction of \textit{strategy-proofness} in all cases.  Specifically, in the aforementioned scenario, when user $i$ cheats a bit on his speedup vector, the benefits for honest users should solely arise from changes in the $x$ and $y$ components (corresponding to two columns in $X$).  Consequently, we have $y_{i+1}'>y_{i+1}$. Furthermore, let us define $x_i'=x_i-\epsilon_x$, and $y_i'=y_i-(y_{i+1}'-y_{i+1})-\epsilon_y$, where $\epsilon_x$ and $\epsilon_y$ represent the change in allocation to benefit honest users, except for user $i+1$. Based on these definitions, we deduce that $a_i\cdot x_i' + b_i\cdot y_i' < a_i\cdot x_i + b_i\cdot y_i$, implying the cunning user must be penalized. This completes the proof.
\end{proof}

\section{Evaluation}
In this section, we present the evaluation results of the OEF system. We begin by exploring the attainment of fairness properties and the enhanced resource efficiency offered by OEF throughout the entire scheduling period, considering both non-cooperative and cooperative scenarios in comparison to various baselines. To illustrate the realization of fairness properties, we conducted a small-scale experiment involving four tenants. Additionally, we employed a large-scale experiment encompassing a more extensive array of tenants and jobs to showcase the advantages in training throughput and job completion time (JCT) reduction. Furthermore, we shed light on the system overhead of OEF, illustrating its seamless integration into large-scale clusters.

\subsection{Experiment Setup}
\label{sec:setup}
\subsubsection{Environment} We conducted experiments on a real cluster consisting of eight RTX 3070 GPUs, eight 3080 GPUs, and eight 3090 GPUs, where four GPUs of the same type co-locate on the same host. Without special mention, we set 5 minutes as the length of the scheduling round by default.

\subsubsection{Workloads} To validate our scheme in real world, we picked several popular DL jobs from different domains. For image classification, we utilized VGG, ResNet, and DenseNet models on the CIFAR-100 dataset. In contrast, for language modeling tasks, we implemented LSTM, RNN, and Transformer models using the WikiText-2 dataset. To ensure consistency with real-world scenarios, each job was assigned a random combination of hyperparameters, including batch size and learning rate, within a reasonable range.

To more accurately mimic real-world scheduling scenarios when assessing JCT and throughput performance, we maintained cluster contention levels consistent with those observed in Microsoft's Philly trace~\cite{jeon2019analysis}.

\subsubsection{Baselines} To demonstrate the superiority of OEF, we compare it to the following schedulers:
\begin{itemize}
%\begin{itemize}[leftmargin=*]
    \item $\bm{\mbox{Gandiva}_{fair}}$~\cite{Gandivafair}. $\mbox{Gandiva}_{fair}$ applies a trading-based algorithm to guarantee \textit{sharing-incentive} and improve the overall resource efficiency. Specifically, it conducts trading in a greedy way, i.e., always trades between shares with the greatest speedup gap. However, the second-price trading mechanism introduced in this user prevents users from obtaining the highest overall efficiency and also incentives users to lie.
    \item \textbf{Gavel}~\cite{gavel}. Gavel formulates a max-min optimization problem, which tries to maximize the throughput of the user with the least speedup. In this sense, Gavel can also provide users with \textit{sharing-incentive}.
\end{itemize}
In order to ensure a fair evaluation, we employ a consistent round-robin scheduling strategy for jobs within each tenant across all baselines. More precisely, scheduling priority is assigned to jobs with the longest starvation time.

\begin{figure}
\begin{minipage}[h]{0.49\linewidth}
\centering
\includegraphics[width=0.99\linewidth]{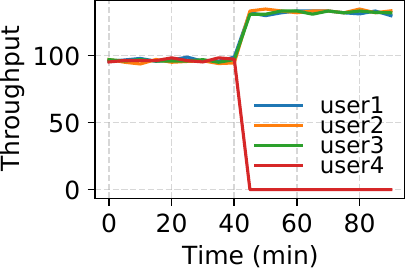}
%\vspace{-1.8em}
\subcaption{No one cheats}
\label{fig:throughput_fluctuation}
\end{minipage}
\begin{minipage}[h]{0.49\linewidth}
\centering
\includegraphics[width=0.99\linewidth]{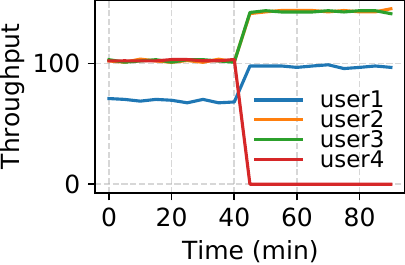}
%\vspace{-1.8em}
\subcaption{User 1 cheats}
\label{fig:throughput_fluctuation_cheating}
\end{minipage}
%\vspace{-.5em}
\caption{OEF penalizes lying users. User 4 exits at the 40-th minute.}
%\vspace{-1em}
\label{fig:sp_experiment}
\end{figure}

\subsubsection{Metrics}
\label{sec:metrics}
To measure fairness and efficiency, we employed normalized throughput as our primary performance metric. This involves normalizing the training throughput of users based on the allocation provided by the scheduler at a specific time unit, using the throughput of the same number of the slowest GPU devices as a reference. 
This normalization process is intended to eliminate the effect of diverse training speeds across different jobs. Furthermore, we performed ablation studies on the evaluation of throughput performance to showcase the individual contributions of the fair share evaluator and placer. We specifically designate the throughput estimate assessed by the OEF fair share evaluator as 'estimated' to demonstrate the algorithmic perspective's throughput advantage. Additionally, we label the tangible training throughput as 'actual', encompassing the throughput benefit from both algorithm design and placement scheme.

\subsection{Evaluation on Fairness Properties}
In this part, we present the attainment of fairness properties under OEF in both cooperative and non-cooperative environments.
%In this part, we present the experimental results of OEF in a non-cooperative environment. In all experiments, four users repeatedly  submit the same type of jobs and each job require one to four GPUs.

\subsubsection{The power of strategy-proofness}
We examined the throughput across users in two scenarios: when no one cheats and when User-1 cheats, while maintaining the same speedup profile. As shown in Fig.~\ref{fig:sp_experiment}(a), when no one cheats on the speedup vector, four users achieve almost identical normalized progress under the non-cooperative OEF, despite some variations due to profiling errors and context switching costs. Moreover, even after the User-4 who executes a batch of VGG11 jobs quits the cluster at the 40th minute, other three users still obtain the same throughout. In Fig.~\ref{fig:sp_experiment}(b), we consider User-1 tuning LSTM jobs artificially increases his speedup rate. In this scenario, we observe that the deceptive user has been penalized, receiving less throughput than before cheating. This demonstrates the effectiveness of \textit{strategy-proofness} under OEF, which prevents a decrease in throughput performance when dealing with shrewd users. At the same time, we also notice that the throughput of honest users improves after someone cheats. Furthermore, it also shows that the throughput disadvantage of User-1 persists even after a user quits the cluster, and cheating reduces the overall training throughput by around 10\%.

\begin{figure}
\begin{minipage}[t]{0.49\linewidth}
\centering
\includegraphics[width=0.99\linewidth]{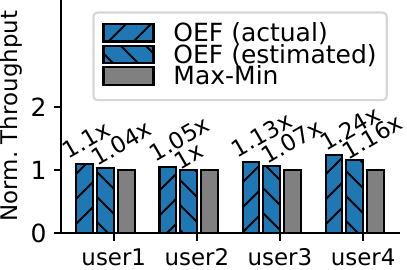}
\subcaption{\textit{Sharing incentive} under cooperative OEF}
\label{fig:si}
\end{minipage}
\begin{minipage}[t]{0.49\linewidth}
\centering
\includegraphics[width=0.99\linewidth]{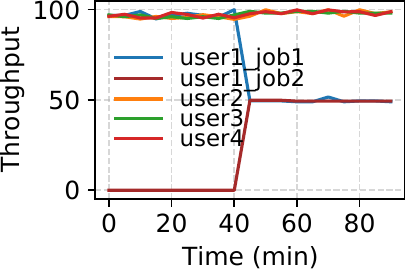}
\subcaption{Multiple types of jobs for user 1 under OEF where User 1 submits one more job at the 40-th minute}
\label{fig:throughput_fluctuation_vary_jobs}
\end{minipage}

%\vspace{-.5em}
\caption{Provision of sharing incentive and adaptation to multiple job types under OEF.}
%\vspace{-1em}
\label{fig:noncooperative_attr}
\end{figure}

\subsubsection{Effectiveness of sharing incentive}
We assessed the \textit{sharing incentive} property of OEF in a cooperative environment by comparing the training throughput of four users when training different types of jobs under OEF. As depicted in Fig.~\ref{fig:noncooperative_attr}(a), the fair share evaluator in OEF consistently produces higher or equal estimated throughput in comparison to the Max-Min approach. Notably, OEF demonstrates its capacity to expedite User-4 the most, achieving up to a $1.16\times$ improvement in estimated training throughput. This acceleration is attributed to the cooperative nature of OEF, which prioritizes the training jobs of tenants with the highest speedup rate, thus maximizing the overall throughput improvement. Furthermore, the placer component in OEF introduces network contention alleviation, resulting in an additional training throughput improvement of $1.24\times$.

\subsubsection{Different speedups within a tenant}
To show the effective support for users to train different types of DL jobs under OEF, we conducted experiments on a dynamic trace, where User-1 adds another type of DL jobs on the fly. As shown in Fig.~\ref{fig:noncooperative_attr}(b), we observe that before introducing the new job type, User-1 receives almost the same throughput as others. After introducing a new type of jobs at 40 minutes, we investigate the throughput of these two types of User-1 separately. It is clear that these two types receive almost equal throughput. Additionally, each job from User-1 obtain half the training throughput of other users.

\subsubsection{Effectiveness of envy freeness}
We proceed to investigate the effectiveness of \textit{envy-freeness} under cooperative OEF. To be specific, we compared the progress obtained by each user to that generated by having resource allocation from others. As illustrated in Fig.~\ref{fig:ef}, users under cooperative OEF are all provided with the most suitable allocation that yields the highest throughput, i.e., no allocation vector from other users results in higher throughput. For instance, User-4 exhibits an estimated throughput that is $1.58\times$ greater than that achieved using the allocation assigned to User-1. To this end, it is unnecessary for users to envy others under the cooperative OEF.

\begin{figure}
\centering
\includegraphics[width=0.96\linewidth]{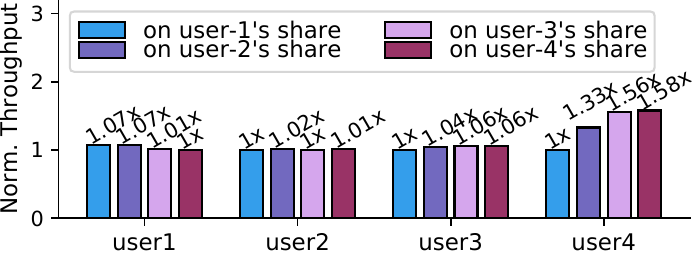}
%\vspace{-.5em}
\caption{Estimated training throughput of allocation for each user from other's perspective under cooperative OEF, which demonstrates that one wouldn't prefer the allocation of others, achieving the \textit{envy freeness}.}
\label{fig:ef}
%\vspace{-1em}
\end{figure}

\begin{figure}
\begin{minipage}[h]{0.49\linewidth}
\centering
\includegraphics[width=0.99\linewidth]{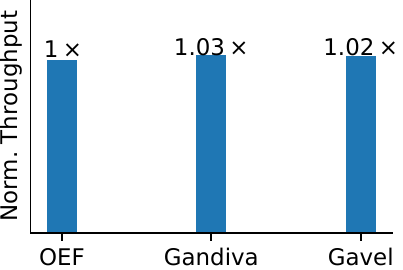}
\subcaption{Estimated throughput}
\label{fig:overall_effi}
\end{minipage}
\begin{minipage}[h]{0.49\linewidth}
\centering
\includegraphics[width=0.99\linewidth]{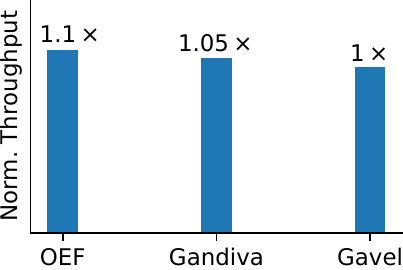}
\subcaption{Actual throughput}
\label{fig:overall_effi_sys}
\end{minipage}
%\vspace{-.5em}
\caption{Training throughput improvement under non-cooperative setting.}
%\vspace{-1em}
\label{fig:noncooperative_throughput}
\end{figure}

\subsection{Resource Efficiency}
\subsubsection{Throughput evaluation}
We conducted an evaluation involving 20 tenants, each owning DL training jobs of the same type. This ensures a fair comparison with Gavel and $\mbox{Gandiva}_{fair}$, which do not support multiple job types within a tenant. As depicted in Fig.~\ref{fig:noncooperative_throughput}, the non-cooperative OEF attains efficiency levels comparable to baselines, while maintaining a \textit{strategy-proofness} property. Through meticulous placement design in the placer, OEF enhances training throughput, achieving up to a 10\% actual increase in training throughput compared to Gavel and $\mbox{Gandiva}_{fair}$, both of which lack optimization strategies for placement, including network contention alleviation and mechanisms to prevent excessive GPU allocation across diverse types.

In Fig.~\ref{fig:cooperative_overall_effi}, the throughput advantage of cooperative OEF over other baselines is evident, primarily attributed to the effective efficiency optimization within the OEF framework. Specifically, OEF outperforms baselines by 20\% in terms of estimated throughput due to the algorithm design. This significant improvement is further amplified to a 32\% actual throughput increase, underscoring the impact of OEF's meticulous placement design.

\begin{figure}[t]
\begin{minipage}{0.49\linewidth}
\centering
\includegraphics[width=0.99\linewidth]{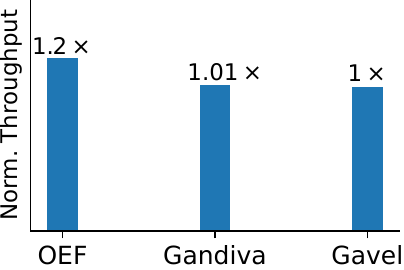}
\subcaption{Estimated throughput}
\end{minipage}
\begin{minipage}{0.49\linewidth}
\centering
\includegraphics[width=0.99\linewidth]{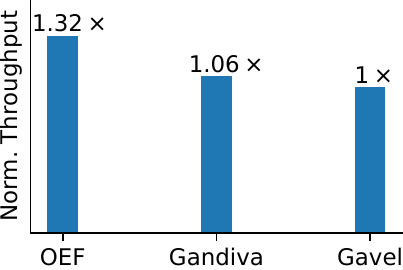}
\subcaption{Actual throughput}
\end{minipage}
%\vspace{-.5em}
\caption{Training throughput improvement under cooperative setting.}
%\vspace{-1em}
\label{fig:cooperative_overall_effi}
\end{figure}

\subsubsection{JCT evaluation}
\label{sec:longtermeval}
We conducted a three-day-long evaluation to showcase the efficacy of reducing JCT in our GPU cluster. The experiment involved jobs from 50 tenants, with each tenant averaging 20 jobs of the same type. Tenants exited the cluster upon the completion of all their jobs. As illustrated in Fig.~\ref{fig:overall_jct}, OEF surpasses its counterparts, Gandiva$_{\mbox{fair}}$ and Gavel, in reducing job training time by 17\% and 19\%, respectively. This JCT advantage is attributed not only to throughput improvements but also to the advanced rounding strategy implemented in the placer. The latter enables tenants to progress jobs even when allocated GPU resources are less than their job expectations, thereby diminishing starvation time and resulting in overall JCT reduction.

\begin{figure}
\centering
\includegraphics[width=0.5\linewidth]{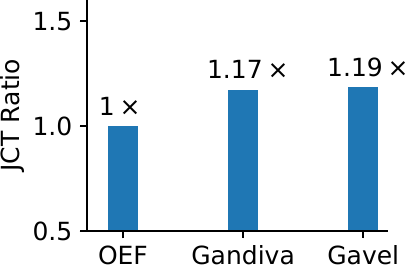}
%\vspace{-.5em}
\caption{OEF can reduce overall JCT by 17\% and 19\%, in contrast to $Gandiva_{fair}$ and Gavel.}
%\vspace{-1em}
\label{fig:overall_jct}
\end{figure}

\subsubsection{Straggler effect alleviation}
We performed an ablation study to demonstrate the extent to which OEF can alleviate the straggler effect. Specifically, we tracked the occurrences of cross-GPU-type placement events across different baselines, where higher-end GPU resources remain idle during periodic data synchronizations with lower-end GPUs. In comparison to Gandiva$_{\mbox{fair}}$ and Gavel, OEF reduces the number of workers impacted by the straggler effect by 14\% and 26\%, respectively.

\subsection{System Overhead}
\subsubsection{Computation overhead}
As scalability is a major concern in large-scale cluster scheduling, we demonstrated that OEF will not block scheduling given a large cluster size. In particular, because solving the optimization problem is of the most time-consuming part under OEF, we carefully investigated its computation overhead at different cluster sizes. Specifically, we fix the number of GPU types to ten, which is sufficient for most production clusters. As shown in Fig.~\ref{fig:sensitivity}(a), we can see that cooperative OEF consumes more time than the non-cooperative version. The reason is that cooperative OEF has $O(n^2)$ constraints while the non-cooperative version has only $O(n)$ constraints, where $n$ is the number of users. Considering that scheduling rounds are several minutes long, this scheduling overhead is negligible.

% \vspace{-.5em}
\subsubsection{Sensitivity analysis}
\label{sec:sensitivity}
As profiling cannot always provide users with strictly accurate runtime information about their jobs, schedulers should be less sensitive to the profiling errors. To this end, we investigated the efficiency deviation, i.e., the gap between the throughput OEF should achieve based on the speedup profile reported by the users and the throughput OEF has actually achieved in the cluster. As shown in Fig.~\ref{fig:sensitivity}(b), even given a profiling error up to 20\%, OEF only yields a deviation about 3\%, demonstrating its robustness in the presence of profiling error.

\begin{figure}[!tb]
\begin{minipage}[!htb]{0.49\linewidth}
\centering
\includegraphics[width=0.99\linewidth]{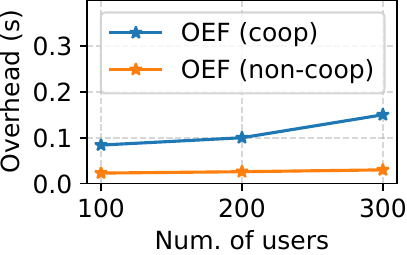}
%\vspace{-1.5em}
\subcaption{Computation overhead}
\end{minipage}
\begin{minipage}[!htb]{0.49\linewidth}
\centering
\includegraphics[width=0.99\linewidth]{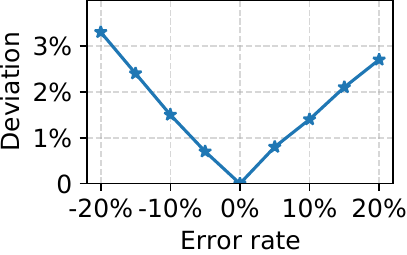}
%\vspace{-1.5em}
\subcaption{Throughput deviation}
\end{minipage}
%\vspace{-1em}
\caption{The scalability of OEF and robustness in the presence of profiling errors.}
%\vspace{-1em}
\label{fig:sensitivity}
\end{figure}

\section{Related Works}
%\noindent\textbf{Homogeneous DL schedulers.} Recently, researchers begin to design efficient job schedulers for DL workloads. According to the ability to support heterogeneous hardware, these schedulers can be divided into two categories, i.e., heterogeneity-agnostic schedulers and heterogeneity-aware schedulers. Gandiva~\cite{gandiva} is a heterogeneity-agnostic scheduler that utilizes the usage pattern of jobs and affinity information to improve scheduling efficiency. Similarly, Tiresias~\cite{tiresias} proposes an efficient algorithm that can schedule DL jobs in a homogeneous cluster without requiring the information of training time. AFS~\cite{AFS} takes another approach to reduce overall JCT, i.e., assigning available GPUs to elastic DL jobs that benefit the overall JCT most. Under a homogeneous cluster, Themis~\cite{themis} can narrow the JCT gap among DL jobs as much as possible, to improve inter-job fairness.

\noindent\textbf{Heterogeneous DL schedulers.} As for heterogeneity-aware scheduling, AlloX~\cite{le2020allox} designs an efficient scheduling algorithm based on bipartite matching to reduce the overall JCT in an offline setting. To systematically deal with the effect introduced by performance heterogeneity, Gavel~\cite{gavel} proposes an optimization framework to yield good efficiency or fairness. By contrast, Gandiva$_{\mbox{fair}}$ focuses on fair scheduling in heterogeneous GPU clusters, and provides several fairness properties. In addition to GPU resources, several works also take other resources such as CPU and memory into consideration when improving job training performance~\cite{looking-beyond, sun2019fair}.

\noindent\textbf{Fair schedulers.}
% Max-Min fairness was proposed by~\cite{jaffe1981bottleneck} to ensure fairness among users, i.e., to provide users with at least $\frac{1}{n}$-th of the total capacity.
Many schedulers have been designed to improve and extend max-min fairness in the literature~\cite{demers1989analysis, parekh1992generalized, zhang1996wf2q, goyal1996start, shreedhar1995efficient, stoica1997hierarchical, baruah1993proportionate, baruah1995fast, zhu2003multiple, blanquer2001fair, kleinberg1999fairness}. These studies primarily focus on the single-resource-type environment, commonly applied in operating system and network-related contexts. As for multi-resource fairness, DRF~\cite{DRF} and its variants~\cite{hug,carbyne,joe2013multiresource,wang2014multi,khamse2017per,wang2016multi} becomes the mainstream manner to provide fairness for tenants. However, all of these schedulers do not well fit into heterogeneous GPU scheduling setting. Themis~\cite{themis} propose a \textit{finish-time} fair scheduler for homogeneous GPU cluster, where jobs are provided \textit{sharing incentive} in best-effort manner. Gandiva$_{\mbox{fair}}$ and Gavel~\cite{Gandivafair,gavel} are the only two schedulers that has similar targets with OEF, with few fairness properties and non-optimal efficiency provided.

% Besides, existing schedulers assume each tenant has the same demand~\cite{DRF, hug}, i.e., the same resource demand vector or correlated resource, similar to our setting.

\section{Conclusion and Remarks}
In this paper, we have identified the inherent conflicts between resource efficiency and various fairness properties in the context of heterogeneous GPU scheduling. Based on this recognition, we design a framework OEF to explore the optimal balance between these conflicting objectives. To the best of our knowledge, OEF is the first solution capable of achieving \textit{strategy-proofness} while utilizing interchangeable resources. Furthermore, OEF not only significantly enhances training throughput but also demonstrates strong scalability for supporting distributed DL training in large-scale clusters.

% This paper designs OEF, a GPU cluster scheduler for achieving fairness and high resource efficiency in the presence of GPU heterogeneity. OEF identifies the conflicts among fairness properties and provides users with fairness either in a cooperative or non-cooperative environment. At its core, OEF uses fine-grained formulations to characterize the desired properties, which is beneficial to combine multiple properties together, improving fairness and efficiency performance. Experiments conducted on cluster and simulation demonstrate OEF significantly outperforms existing heterogeneity-aware baselines.

OEF can be extended to support job-level fairness.
% , specially in scenarios where the number of GPUs allocated to a job may fluctuate over time. 
With the increasing popularity of elastic DL training~\cite{AFS}, OEF can leverage this training paradigm to dynamically assign GPU resources to each job, while simultaneously ensuring fairness in the allocation process.
% while there remains several challenges. The most significant concern is that the GPU share of jobs under the OEF's framework is elastic, i.e., the number of GPUs a job obtains may vary over time. Though elastic DL training is common, an immense performance gap between two given GPU types may invalidate the efficiency gaining, denoting that given extra GPUs jobs can even be slowed down. As demonstrated in \ref{theo:adjacentnonzero}, OEF can somehow bypass this dilemma since it will only allocate adjacent GPU types to each user, i.e., GPUs with insignificant performance diversity.

% In this sense, a similar skewed batch size partition can be adopted to mitigate the influence and maintain efficiency gaining.

\begin{acks}
We thank the anonymous reviewers, and our shepherd, Prof. Robert Birke for their helpful comments. This work is supported in part by the Science and Technology Development Fund of Macau (0024/2022/A1, 0071/2023/ITP2), the Multi-Year Research Grant of University of Macau (MYRG2022-00119-FST, MYRG-GRG2023-00019-FST-UMDF).
\end{acks}

\bibliographystyle{plain}
\bibliography{OEF}


\end{document}